\documentclass[conference,letterpaper]{IEEEtran}
\IEEEoverridecommandlockouts 

\usepackage[letterpaper,%
            left=.64in,right=.64in,top= 0.64in,bottom= 0.64in]{geometry}

\usepackage{amsfonts}                                                          
\usepackage{epsfig}

\usepackage{amsthm}
\usepackage{graphicx}        
\usepackage{latexsym}
\usepackage{amssymb}
\usepackage{amsmath}
\usepackage{parskip}
\usepackage{multirow}
\usepackage{cite}
\usepackage{mathtools}
\usepackage{epstopdf}
\usepackage{etoolbox}
\usepackage{array}
\usepackage{mathptmx}
\usepackage[bottom]{footmisc}
\usepackage{tikz}

\usetikzlibrary{decorations.pathreplacing}

\usepackage{verbatim}
\usepackage{calrsfs}
\usepackage{booktabs}

\DeclareMathAlphabet{\pazocal}{OMS}{zplm}{m}{n}

\newcommand{\mb}{\mathbb}

\let\bbordermatrix\bordermatrix
\patchcmd{\bbordermatrix}{8.75}{4.75}{}{}
\patchcmd{\bbordermatrix}{\left(}{\left[}{}{}
\patchcmd{\bbordermatrix}{\right)}{\right]}{}{}

\newcommand{\sr}{\stackrel}

\newcommand{\rar}{\rightarrow}

\newcommand{\tri}{\sr{\triangle}{=}}

\newcommand{\be}{\begin{equation}}
\newcommand{\ee}{\end{equation}}
\newcommand{\bea}{\begin{eqnarray}}
\newcommand{\eea}{\end{eqnarray}}
\newcommand{\bes}{\begin{eqnarray*}}
\newcommand{\ees}{\end{eqnarray*}}
\newcommand{\bce}{\begin{center}}
\newcommand{\ece}{\end{center}}
\newcommand{\beae}{\begin{IEEEeqnarray}{rCl}}
\newcommand{\eeae}{\end{IEEEeqnarray}}

\def\VR{\kern-\arraycolsep\strut\vrule &\kern-\arraycolsep}
\def\vr{\kern-\arraycolsep & \kern-\arraycolsep}

\newcommand{\ben}{\begin{enumerate}}
\newcommand{\een}{\end{enumerate}}

\newcommand{\hso}{\hspace{.1in}}
\newcommand{\hst}{\hspace{.2in}}

\newtheorem{theorem}{Theorem}[section]

\newtheorem{remark}{Remark}[section]
\newtheorem{corollary}{Corollary}[section]

\newtheorem{assumptions}{Assumptions}[section]
\newtheorem{definition}{Definition}[section]

\begin{document}



\title{A Riccati-Lyapunov Approach to Nonfeedback Capacity of MIMO Gaussian Channels Driven by  Stable and Unstable Noise\\
}
\author{\IEEEauthorblockN{ Charalambos D. Charalambous and Stelios Louka}
\IEEEauthorblockA{
\textit{University of Cyprus}\\
chadcha@ucy.ac.cy,louka.stelios@ucy.ac.cy}
}

\maketitle

\begin{abstract}
In this paper it is shown that the nonfeedback capacity of  multiple-input multiple-output (MIMO) additive Gaussian noise (AGN) channels, when the noise is nonstationary and unstable, is characterized by an asymptotic optimization problem that involves,  a generalized matrix algebraic  Riccati  equation (ARE) of filtering theory,  and a matrix Lyapunov equation of stability theory of  Gaussian systems. Furthermore, conditions are identified such that,  the  characterization of nonfeedback capacity corresponds to the uniform asymptotic per unit time limit, over all  initial distributions,  of  the characterization of a  finite block or transmission  without feedback information (FTwFI) capacity, which involves,  two  generalized matrix difference   Riccati  equations (DREs) and a  matrix difference Lyapunov equation. 

\end{abstract}

%
\section{Introduction, Problem,  and Main Results}
\label{sect:problem}
\par Shannon's information theoretic  definition of nonfeedback capacity of additive Gaussian noise (AGN) channels with memory,  is a fundamental  mathematical tool for  analysis and synthesis of  reliable data transmission over noisy communication channels. The characterization of  nonfeedback capacity in frequency-domain, for stationary or asymptotically stationary channels, i.e., when the channel noise inputs and outputs are asymptotically stationary,  gave rise to the so-called water-filling solution, which is documented in  \cite{gallager1968,ihara1993,cover-thomas2006,yeung2008} and in several research papers, such as Tsybakov \cite{tsybakov2006}. The analysis of   channel capacity for   asymptotically equivalent  matrices for  MIMO Gaussian channels,   is found in  \cite{brandenburg-wyner:1974,gutierrez-crespo:2008} and more recently in \cite{gutierrez-crespo-rodriguez-hogstad:2017}.  For Gaussian channels with intersymbol   interference in  \cite{hirt-massey:1988}. 

Bounds on nonfeedback capacity for single-input single-output (SISO)  AGN channel with stable noise with memory,  are derived in  \cite{yanaki1992,yanaki1994,chen-yanaki1999}, as well as comparisons to feedback capacity. 
 Recently,  sequential time-domain characterizations of nonfeedback capacity for SISO  AGN channels with unstable,  finite-memory autoregressive noise are presented  in \cite{kourtellaris-charalambous-loyka:2020a}.  Sequential  characterization of nonfeedback capacity for SISO AGN channels with general unstable noise with memory,    are obtained in  \cite{charalambous2020new}, and lower bounds in   \cite[Corollary~II.1]{charalambous-kourtellaris-loykaIEEEITC2019}. Additional lower bounds are discussed in 
  \cite{louka-kourtellaris-charalambous:2020b}, \cite{kourtellaris-charalambous-loyka:2020b}, which are equivalent to the  Cover and Pombra \cite{cover-pombra1989} nonfeeback capacity formula. Equivalent characterizations of Cover and Pombra $n-$finite transmission,  or block length for MIMO AGN channels with memory  are presented in \cite{charalambous-kourtellaris-louka:2020a}. 

The purpose of this paper is to derive new results on nonfeedback capacity for  multiple-input multiple-output (MIMO)   AGN channels, driven by unstable,  nonstationary and nonergodic  noise in time-domain,  described by 
\begin{align}
Y_t=H_tX_t+V_t,  \ \  t=1, \ldots, n, \ \ \frac{1}{n}  {\bf E} \Big\{\sum_{t=1}^{n} ||X_t||_{{\mathbb R}^{n_x}}^2\Big\} \leq \kappa \label{g_cp_1} 
\end{align}
where   $\kappa \in [0,\infty)$,   $X_t :  \Omega \rar  {\mathbb X}\tri  {\mathbb R}^{n_x}$,   $Y_t :  \Omega \rar  {\mathbb Y}\tri  {\mathbb R}^{n_y}$,  and   $V_t :  \Omega \rar {\mathbb V}\tri  {\mathbb R}^{n_y}$,  are the channel input, channel output and noise random variables (RVs),  respectively,  $(n_x, n_y)$ are finite positive integers,  $H_t \in {\mathbb R}^{n_y\times n_x}$ is nonrandom and the distribution of the  sequence $V^n = \{ V_1, \ldots, V_n\}$,  i.e.,   ${\bf P}_{V^n}\tri {\mathbb P}\{V_1\leq v_1, \ldots, V_n\leq v_n\}$, is jointly Gaussian. \\
The main fundamental difference from previous characterizations of nonfeedback capacity found in the literature, i.e.,  \cite{gallager1968,ihara1993,cover-thomas2006,yeung2008,tsybakov2006}, is that the consideration of nonergodic, time-varying and unstable noise, gives higher  achievable rates compared to stable noise (see \cite{charalambous2020new,louka-kourtellaris-charalambous:2020b,kourtellaris-charalambous-loyka:2020b,kourtellaris-charalambous-loyka:2020a}). 

{\it Operational Nonfeedback Code.} The code for the MIMO AGN channel (\ref{g_cp_1}), is denoted by  $\{(n, {\cal M}^{(n)}, \epsilon_n, \kappa):n=0, 1, \dots\}$ and consists of  (a) a  set of uniformly distributed  messages $M: \Omega \rar {\cal M}^{(n)} \tri \{ 1,  \ldots, M^{(n)}\}$,  
known to the encoder and decoder, (b) 
a set of  encoder strategies mapping messages $M=m$ and past channel inputs into current inputs, defined by\footnote{The superscript on expectation operator ${\bf` E}^g$  indicates that the corresponding distribution ${\bf P}= {\bf P}^g$  depends on the encoding strategy $g$.} 
\begin{align}
&{\cal E}_n(\kappa) \triangleq  \Big\{g_i: {\cal M}^{(n)}  \times   {\mathbb X}^{i-1}   \rar {\mb X}_i, \:x_1=g_1(m), \;
x_2=g_2(m,x_1), \nonumber \\
& \ldots,  x_n=g_n(m,x^{n-1})\; \Big|\;   \frac{1}{n}   {\bf E}^g    \Big\{\sum_{t=1}^n ||X_t||_{{\mathbb R}^{n_x}}^2\Big\} \leq \kappa \Big\} \label{NCM-6}
\end{align}
where   $g_i(\cdot)$ are measurable maps and (c) 
a decoder $d_{n}(\cdot):{\mb Y}^n\rar {\cal M}^{(n)}$,  with average probability of decoding error 
\begin{align}
{\bf P}_{error}^{(n)} \triangleq& \frac{1}{M^{(n)}} \sum_{m \in {\cal M}^{(n)}} {\bf  P}^g \Big\{d_{n}(Y^{n}) \neq m \big|  M=m\Big\}
\leq \epsilon_n.  \label{NCM-7} 
\end{align}
Note that ${\bf P}_{error}^{(n)}$ depends on the distribution of $V_1$, i.e., ${\bf P}_{V_1}$.
The messages $M: \Omega \rar {\cal M}^{(n)}$ are assumed independent of the noise sequences $V^n$,  that is $
{\bf P}_{V^n|M}={\bf P}_{V^n}$. 
{\it The code rate} is defined by $r_n\triangleq \frac{1}{n} \log M^{(n)}$.  A rate $R$ is called an {\it achievable rate}, if there exists an encoder and decoder sequence satisfying
$\lim_{n\longrightarrow\infty} {\epsilon}_n=0$ and $\liminf_{n \longrightarrow\infty}\frac{1}{n}\log{M}^{(n)}\geq R$. The operational definition of the {\it nonfeedback capacity} is $C^{OP}(\kappa)\triangleq \sup \{R\big| R \: \: \mbox{is achievable}\}\; \forall\; {\bf P}_{V_1} $, i.e., it does not depend on ${\bf P}_{V_1}$.

\subsection{Main Results of the Paper}
Throughout this paper, we consider the noise of Definition~\ref{def_nr_2}.

\begin{definition} 
 \label{def_nr_2}
A time-varying partially observable  state space (PO-SS)  realization of  the Gaussian noise $V^n$,   is defined by  
\begin{align}
&S_{t+1}=A_t S_{t}+ B_t W_t, \hso t=1, \ldots, n-1\label{real_1a}\\
&V_t= C_t S_{t} +  N_t W_t,  \hso t=1, \ldots, n, \label{real_1_ab}\\
 & S_1\in G(\mu_{S_1},K_{S_1}), \hso K_{S_1} \succeq 0, \\
&W_t\in G(0,K_{W_t}),\hso K_{W_t} \succ 0, \hso  t=1 \ldots, n, \\
& S_t : \Omega \rar {\mathbb R}^{n_s}, \  W_t : \Omega \rar {\mathbb R}^{n_w}, \  V_t : \Omega \rar {\mathbb R}^{n_y}, \\\ &R_t\tri N_t K_{W_t} N_t^T \succ 0,\hso t=1, \ldots, n \label{cp_e_ar2_s1_a_new}
\end{align}
where  $W_t,   t=1 \ldots, n$ is an independent Gaussian process,  independent of $S_1$, $n_y, n_s, n_w$ are arbitrary positive integers,  $(A_t, B_t, C_t, N_t, K_{S_1}, K_{W_t})$ are nonrandom, $X \in G(\mu_X, K_X)$ means $X$ is Gaussian distributed with mean $\mu_X$ and covariance  $K_X$ and for any matrix $Q \in {\mathbb R}^{n \times n}$. The notation $Q\succeq 0$ (resp. $Q \succ 0$) means $Q$ is symmetric positive semidefinite (resp. definite).
\end{definition}

{\it Converse Coding Theorem.} Suppose there exists a sequence of achievable nonfeedback codes with error probability ${\bf P}_{error}^{(n)}\rightarrow  0$, as $n \rightarrow \infty$. Then  $
R \leq \lim_{n \longrightarrow \infty}\frac{1}{n} C_n(\kappa, {\bf P}_{Y_1})$, 
where  $C_n(\kappa, {\bf P}_{Y_1})$ is  the sequential characterization of  the $n-$finite block length,  or transmission without feedback information ($n$-FTwFI) capacity formula  \cite[Section~I, III]{charalambous-kourtellaris-louka:2020a}, given as follows. 
\begin{align}
C_n(\kappa, {\bf P}_{Y_1}&)= \sup_{ \frac{1}{n}  {\bf E} \big\{\sum_{t=1}^{n}  ||X_t||_{{\mathbb R}^{n_x}}^2\big\} \leq \kappa    }\sum_{t=1}^n I(X_t, V^{t-1};Y_t| Y^{t-1})\label{seq_2_nfb}\\
=&\sup_{ \frac{1}{n}  {\bf E} \big\{\sum_{t=1}^{n}  ||X_t||_{{\mathbb R}^{n_x}}^2\big\} \leq \kappa    } H(Y^n)-H(V^n) \in [0,\infty] \label{seq_2_nfb_a}
\end{align}
where (\ref{seq_2_nfb_a}) follows from  the channel definition (\ref{g_cp_1}) (provided the probability density functions exist) and 
 the supremum is over ${\bf P}_{X_t|X^{t-1}}, t=1, \ldots, n$  induced by, 
\begin{align}
&X_t = \Lambda_t {\bf X}^{t-1} + Z_t,    \hso X_1=Z_1,  \label{ch_in_1} \\
&{\bf X}^n\tri \left[ \begin{array}{cccc} X_1^T &X_2^T &\ldots &X_n^T\end{array}\right]^T,  \\
&Z_t\in  G(0, K_{Z_t}), \; K_{Z_t}\succeq 0,  \; t=1, \ldots, n,   \mbox{ indep.  Gaussian}, \\
& Z_t \; \mbox{independent of} \;  (V^{t-1},X^{t-1},Y^{t-1},Z^{t-1}),  t=1, \ldots, n, \\
&{\Lambda}_t \in {\mathbb R}^{n_x\times (t-1)n_x}, \; \mbox{$ \forall t$  is  nonrandom}.\label{ch_in_2}
\end{align}
By recursive substitution of ${\bf X}^{t-1}$ into the right hand side of (\ref{ch_in_1}), we obtain  $X_t=\overline{Z}_t$, where $\overline{Z}_t \in G(0, K_{\overline{Z}_t}), t=1, \ldots, n$, is a correlated Gaussian noise, as given in  \cite{cover-pombra1989}.  However, for the purpose of our asymptotic analysis,  we prefer (\ref{ch_in_1})-(\ref{ch_in_2}), because it is much easier to analyse.  

We emphasize that the consideration of unstable noise $V^n$ implies $Y^n$ is also unstable and therefore for our asymptotic analysis, we need to use the two innovations processes of $V^n$ and  $Y^n$, as in \cite{charalambous2020new,louka-kourtellaris-charalambous:2020b,kourtellaris-charalambous-loyka:2020b,kourtellaris-charalambous-loyka:2020a}, giving rise to the following characterization of  ${C}_{n}(\kappa, {\bf P}_{Y_1})\in [0,\infty]$ given by, \\
\begin{align}
{C}_{n}(\kappa, {\bf P}_{Y_1})
= \sup_{(\Lambda_t, K_{\overline{Z}_t}), t=1, \ldots, n,  \frac{1}{n}  {\bf E} \big\{\sum_{t=1}^{n} ||X_t||_{ {\mathbb R}^{n_x}}^2\big\} \leq \kappa    }\sum_{t=1}^n \Big\{ H(I_t)-H(\hat{I}_t)\Big\} \label{ftfic_is_g_in_nfb}
\end{align}
where $I_t, \hat{I}_t$ are the innovations processes of $Y^n, V^n$, 
\begin{align}
I_t \tri Y_t- {\bf E}\big\{Y_t\Big|Y^{t-1}\big\}, \hst  \hat{I}_t \tri V_t- {\bf E}\big\{V_t\Big|V^{t-1}\big\}. \label{inn_intr_thm1_nfb}
\end{align}
Clearly, the convergence properties of  $\lim_{n \longrightarrow \infty}  \frac{1}{n}  C_n(\kappa,{\bf P}_{Y_1})$ and its independent of all ${\bf P}_{Y_1}$, are directly related to the convergence properties of $(I_t, \hat{I}_t, X_n), t=1,2, \ldots, n$, as $n \rightarrow \infty$.  

{\it State Space Realization of Channel Input.} For the analysis of the limit $\lim_{n \longrightarrow \infty}  \frac{1}{n}  C_n(\kappa,{\bf P}_{Y_1})$, we consider the alternative, state space realization of the Gaussian input $X^n$. 
\begin{align}
&\Xi_{t+1}=F_{t} \Xi_{t}+ G_{t} Z_t, \hso t=1, \ldots, n-1,\label{real_1aa}\\
&X_t= \Gamma_t \Xi_{t} +  D_t Z_t,  \hso t=1, \ldots, n, \label{real_1_abb}\\
 & \Xi_1\in G(\mu_{\Xi_1},K_{\Xi_1}), \; K_{\Xi_1} \succeq 0, \\
&Z_t\in G(0,K_{Z_t}),\hso K_{Z_t} \succeq 0, \hso  t=1 \ldots, n,  \\
&\mbox{$Z^n$ indep. seq.}, \hso \mbox{$(\Xi_1, Z^n, W^n)$ mutually indep.}  \\
& \Xi_t : \Omega \rar {\mathbb R}^{n_\xi}, \  Z_t : \Omega \rar {\mathbb R}^{n_z}, \  X_t : \Omega \rar {\mathbb R}^{n_x}
\label{real_1aaa}
\end{align}
where $n_\xi, n_z$ are arbitrary positive integers and $(F_{t}, G_{t}, \Gamma_t, D_t, K_{\Xi_1}, K_{Z_t})$ are nonrandom matrices $\forall t$. \\
Note that any finite-memory $AR$ input $X_t = \sum_{j=1}^M \Lambda_{t,j} X_{t-j} + Z_t$, with arbitrary large $M$, approximates (\ref{ch_in_1}), and this is a special case of (\ref{real_1aa})-(\ref{real_1aaa}). The performance of such inputs with $M=1$ discussed in \cite[Section~IV]{kourtellaris-charalambous-loyka:2020a} and \cite[Theorem~III.3]{louka-kourtellaris-charalambous:2020b} for IID input. 
{\it Asymptotic Characterization of Capacity.} The purpose of this paper is to analyze the two problems listed below.  \\
{\bf Problem \#1.} Identify conditions, such that asymptotic limit exists and does not depend on ${\bf P}_{Y_1}$, 
\vspace{-0.1cm}
\begin{align}
C^{o}(\kappa,{\bf P}_{Y_1})\tri   \lim_{n \longrightarrow \infty}  \frac{1}{n}  C_n(\kappa,{\bf P}_{Y_1})
=C(\kappa)\in [0,\infty), \hso \forall {\bf P}_{Y_1} \label{inter_change_limit_a}
\end{align}
and characterize $C(\kappa)$, which is independent of ${\bf P}_{Y_1}$.   
\begin{assumptions}
Considered for the asymptotic analysis are the two cases of noise and 
channel input realizations.  \\
{\bf Case 1:} Time-invariant, 
\begin{align}
&(A_n, B_n, C_n, N_n, K_{W_n})=(A, B, C, N, K_{W}),\; K_{W}\succ 0,  \label{tic_1}  \hso  \forall n \\ 
&(F_n,G_n,  \Gamma_n, D_n, K_{Z_n})=(F,G,   \Gamma, D, K_{Z}),  K_Z \succeq 0, \hso \forall n.  \label{tic_2}
\end{align}
{\bf Case 2:} Asymptotically time-invariant, 
\begin{align}
 &\lim_{n\longrightarrow \infty}(A_n, B_n, C_n, N_n, K_{W_n})=(A, B, C, N, K_{W}), \: K_W\succ 0, \label{atic_1} \\
 &  \lim_{n\longrightarrow \infty}(F_n,G_n, \Gamma_n, D_n, K_{Z_n})=(F,G,   \Gamma, D, K_{Z}), \: K_Z \succeq 0\label{atic_2}
 \end{align}
 where the limits are element wise.  For both Cases 1 and 2, the time-invariant realizations are assumed of  minimal dimensions.
 \end{assumptions}
For Cases 1 and 2,  we identify conditions on 
1) channel model matrices $(H, A, B, C, N, K_{W})$ and 
2) channel input matrices $(F,G,  K_{Z}, \Gamma, D), K_Z \succeq 0$, 
such that the limit in (\ref{inter_change_limit_a}) exists,  is independent of ${\bf P}_{Y_1}$ and is  characterized by 
\begin{align}
&C^{o}(\kappa,{\bf P}_{Y_1})= C(\kappa)  \tri   \sup \frac{1}{2} \ln \big\{\frac{ \det \big( {\bf C} \Pi {\bf C}^T + {\bf D} K_{\overline{W}}{\bf D}^T \big)    }{\det\big(C \Sigma C^T +N K_{W} N^T\big)}\big\}^+\label{ll_3_in}\\
&\mbox{the supremum is over} \;  (F, G, \Gamma, D, K_{Z}) \; \; \mbox{and} \nonumber\\
&K_Z\succeq 0, \;  tr(\Gamma P \Gamma^T +D K_Z D^T) \leq \kappa, \\
& \Pi \succeq 0, \; \Sigma \succeq 0 \; \mbox{satisfy matrix Algebraic Riccati Eqns}, \\
&P \succeq 0 \hso \mbox{satisfies  a matrix Lyapunov Equation}
 \label{ll_4}
 \end{align}
and where $\{\cdot\}^+\tri \max \{1, \cdot\}$, the $({\bf C}, {\bf D}, K_{\overline{W}})$ are specific matrices, related to the channel model and channel input matrices. 

{\bf Problem \#2.} Under the conditions of Problem \#1, we also show that the limit and the supremum can be interchanged and 
\begin{align}
C^{\infty}(\kappa, {\bf P}_{Y_1}) \tri &  \sup_{\lim_{n \longrightarrow \infty} \frac{1}{n}  {\bf E} \big\{\sum_{t=1}^{n} ||X_t||_{ {\mathbb R}^{n_x}}^2\big\} \leq \kappa    } \lim_{n \longrightarrow \infty} \frac{1}{n}\Big( \sum_{t=1}^n  H(I_t)-H(\hat{I}_t)\Big)      \label{inter_in}\\
 =& C^{o}(\kappa,{\bf P}_{Y_1})= C(\kappa)\in [0,\infty), \hso \forall {\bf P}_Y. 
 \end{align}
 
 {\it Direct Coding Theorem.}  By (\ref{ll_3_in}) and (\ref{inter_in}), the convergence is uniform over all ${\bf P}_{Y_1}$. Hence,  the asymptotic equipartition (AEP) and the information stability hold,  from which follows directly that $C(\kappa)$ is the nonfeedback capacity,  even for unstable channels, similar to the feedback capacity in \cite{kourtellaris-charalambousIT2015_Part_1,charalambous-kourtellaris-loykaIT2015_Part_2}.

{\it Notation.} \\
${\mathbb Z}_+ \tri \{1,2, \ldots\},    {\mathbb Z}_+^n \tri \{1,2, \ldots, n\}$, where $n$ is a finite positive integer. 
${\mathbb R}\tri (-\infty, \infty)$,  and ${\mathbb R}^m$ is the vector space of tuples
of the real numbers for an integer $m\in {\mathbb Z}_+$.  ${\mathbb R}^{n \times m}$ is the set of $n$ by $m$ matrices with entries from the set of real numbers for $(n,m)\in {\mathbb Z}_+\times {\mathbb Z}_+$. $I_{n} \in {\mathbb R}^{n\times n}, n\in {\mathbb Z}_{+}$ denotes the identity matrix, $tr\big(A\big)$ denotes the trace of any matrix $A \in {\mathbb R}^{n\times n}, n\in {\mathbb Z}_+$. \\
${\mathbb C} \tri \{a+j b: (a,b) \in {\mathbb R} \times {\mathbb R}\}$ is the  space of complex numbers. \\
${\mathbb D}_o \tri \big\{c \in {\mathbb C}: |c| <1\big\}$ is the open unit disc of the space of complex number ${\mathbb C}$. 
$spec(A) \subset {\mathbb C}$ is the spectrum of a matrix $A \in {\mathbb R}^{q \times q}, q \in {\mathbb Z}_+$ (the set of all its eigenvalues).
A  matrix $A \in {\mathbb R}^{q \times q}$ is called exponentially stable if all its eigenvalues are within the open unit disc, that is,  $spec(A) \subset {\mathbb D}_o$.\\
$X\in G(\mu_{X}, K_{X}), K_{X}\succeq 0$ denotes a Gaussian distributed RV $X$, with   mean  $\mu_{X}={\bf E}\{X\}$ and covariance $K_{X}\succeq 0$, $
K_X = cov(X,X) \tri {\bf E}\big\{\big(X-{\bf E}\big\{X\big\}\big) \big(X-{\bf E}\big\{X\big\}\big)^T \big\}$.
Given another Gaussian RV $Y: \Omega \rar {\mathbb R}^{n_y}$, which is jointly Gaussian distributed with $X$, i.e., with joint distribution ${\bf P}_{X,Y}$,  the conditional covariance of $X$ given $Y$ is (by properties of Gaussian RVs)
\begin{align}
K_{X|Y} = cov(X,X\Big|Y) \tri
{\bf E}\big\{\big(X-{\bf E}\big\{X\Big|Y\big\}\big) \big(X-{\bf E}\big\{X\Big|Y\big\}\big)^T \big\}.\nonumber
\end{align}
\section{Asymptotic Characterization of Capacity}
 \subsection{Sequential Characterizations of $n-$FTwFI Capacity}
  \label{sect:AGN}
We recall    the characterization of  $C_n(\kappa, {\bf P}_{Y_1})$ of  (\ref{ftfic_is_g_in_nfb}), which shows its dependence on two DREs and a  Lyapunov  equation.
\begin{theorem} \cite[Theorem III.2]{charalambous-kourtellaris-louka:2020a} Sequential  characterization of  ${C}_n(\kappa,{\bf P}_{Y_1})$.
\label{thm_SS_nfb}
  Consider the MIMO AGN channel (\ref{g_cp_1}), the noise of Definition~\ref{def_nr_2} and the input (\ref{real_1aa})-(\ref{real_1aaa}).
Define
\begin{align*}
&\Theta_t \tri   \left(\begin{array}{cc} \Xi_t^T & S_{t}^T\end{array}\right)^T, \hso  \overline{W}_t \tri   \left(\begin{array}{cc} Z_t^T & W_{t}^T\end{array}\right)^T, \hso   \overline{W}_t \in G(0,K_{\overline{W}_t}) , \\
&\Pi_{t} \tri  cov\big(\Theta_t,\Theta_t\Big|Y^{t-1}\big)= {\bf E}\Big\{\Big(\Theta- \widehat{\Theta}_{t}\Big)\Big(\Theta_t-\widehat{\Theta}_{t} \Big)^T\Big\}, \\ 
& \widehat{\Theta}_{t} \tri {\bf E}\Big\{\Theta_t\Big|Y^{t-1}\Big\},  \;  t=2, \ldots, n, \; 
\widehat{\Theta}_1 \tri \mu_{\Theta_1}, \; \Pi_1 \tri K_{\Theta_1} , \\
&P_t\tri cov\big(\Xi_t,\Xi_t)= {\bf E}\Big\{\Big(\Xi_t - {\bf E}\Big\{\Xi_t\Big\}\Big)\Big(\Xi_t - {\bf E}\Big\{\Xi_t\Big\}\Big)^T\Big\}.
\end{align*}
(i) The joint Gaussian process $(X^n,Y^n, V^n)$ is represented by 
\begin{align}
&\Theta_{t+1}= {\bf A}_t \Theta_t + {\bf B}_t \overline{W}_t,  \hso t=1, \ldots, n-1,\\
&Y_t={\bf C}_t \Theta_t  +  {\bf D}_t \overline{W}_t, \hso t=1, \ldots, n
\\
&{\bf A}_t \tri \left( \begin{array}{cc}  F_t & 0 \\ 0 & A_t \end{array} \right), \hso {\bf B}_t \tri \left( \begin{array}{cc}  G_t & 0 \\ 0 & B_t \end{array} \right), \\
& {\bf C}_t \tri \left( \begin{array}{cc}  H_t \Gamma_t & C_t \end{array} \right), \hso {\bf D}_t \tri \left( \begin{array}{cc}  H_t D_t  & N_t \end{array} \right).
\end{align}
where ${\bf A}_t, {\bf B}_t, {\bf C}_t, {\bf D}_t$ are appropriate matrices.\\
(ii) The covariance of the error, $\Pi_t \tri {\bf E} \big\{\widehat{E}_t \widehat{E}_t^T\big\}$,$\widehat{E}_t =\Theta_t-\widehat{\Theta}_t$,  satisfies the generalized matrix DRE 
\begin{align}
&\Pi_{t+1}= {\bf A}_{t} \Pi_{t}{\bf A}_{t}^T  + {\bf B}_{t}K_{\overline{W}_{t}}{\bf B}_{t}^T -\Big({\bf A}_{t}  \Pi_{t}{\bf C}_{t}^T+{\bf B}_{t}K_{\overline{W}_{t}}{\bf D}_{t}^T  \Big) \nonumber \\
& \hspace{0.1cm} .\Big({\bf D}_{t} K_{\overline{W}_{t}} {\bf D}_{t}^T+{\bf C}_{t}  \Pi_{t} {\bf C}_{t}^T\Big)^{-1}\Big( {\bf A}_{t}  \Pi_{t}{\bf C}_{t}^T+ {\bf B}_{t} K_{\overline{W}_{t}}{\bf D}_{t}^T  \Big)^T,\nonumber\\   & \hspace{0.1cm} \hso \Pi_t \succeq 0,  \hso \Pi_{1}=K_{\Theta_1}\succeq 0, \hso t=1, \ldots, n.\label{DREy}
\end{align}
Moreover, the error $\widehat{E}_t \tri {\Theta}_t - \widehat{\Theta}_t$ satisfies the recursion
\begin{align}
&\widehat{E}_{t+1} = {\bf F}^{CL}_t (\Pi_t)\widehat{E}_t + \big({\bf B}_t - {\bf F}_t(\Pi_t){\bf D}_t\big) \overline{W}_t,\; t = 1,\dots, n, \label{error2}\\
& {\bf F}^{CL}_t(\Pi_t) =  {\bf A}_t - {\bf F}_t(\Pi_t){\bf C}_t, \\
&{\bf F}_t(\Pi_t) = \big ({\bf A}_t\Pi_t{\bf C}_t^T +{\bf B}_t K_{\overline{W}_t}{\bf D}_t^T \big) \big({\bf D}_t K_{\overline{W}_t} {\bf D}_t^T + {\bf C}_t \Pi_t {\bf C}_t^T \big)^{-1}.\nonumber 
\end{align}
(iii) The innovations process $I_t$  of $Y^n$ for $t=1, \ldots, n$, is 
\begin{align}
&I_t \tri  Y_t -{\bf E} \Big\{Y_t\Big|Y^{t-1}\Big\}= {\bf C}_t \big(\Theta_t-\widehat{\Theta}_{t}\big)+ {\bf D}_t \overline{W}_t, \label{kf_m_2_nfb_a} \\
& I_t\in G(0, K_{I_t}), \hso K_{I_t} ={\bf C}_t \Pi_t {\bf C}_t^T + {\bf D}_t K_{\overline{W}_t}{\bf D}_t^T.
\end{align}
(iv) The matrix $P_t =  cov\big(\Xi_t,\Xi_t)$ satisfies Lyapunov recursion,  
\begin{align}
&P_{t+1} = F_t P_t F_t^T + G_t K_{Z_t}G_t^T, \hso P_t\succeq 0, \hso P_1=K_{\Xi_1} \label{lyapunov}.
\end{align}
(v) The average power constraint is 
\begin{align}
&\frac{1}{n}  {\bf E} \Big\{\sum_{t=1}^{n} ||X_t||_{{\mathbb R}^{n_x}}^2\Big\}= \frac{1}{n} \sum_{t=1}^{n}  tr \Big(\Gamma_t P_t \Gamma_t^T +D_t K_{Z_t}D_t^T\Big)   \label{cp_9_al_n_nfb_a}.
\end{align}
(vii) The entropy of $Y^n$ is $H(Y^n)=\sum_{t=1}^n H({I}_t)$, is given by 
\begin{align}
H(Y^n)
=   \frac{1}{2}\sum_{t=1}^n \ln \big( (2\pi e)^{n_y} \det\big({\bf C}_t \Pi_t {\bf C}_t^T + {\bf D}_t K_{\overline{W}_t}{\bf D}_t^T\big) \big) \label{entr_output}
\end{align}
and the entropy of $V^n$ is $H(V^n)=\sum_{t=1}^n H(\hat{I}_t)$,  is given by 
\begin{align}
&H(V^n)
=   \frac{1}{2}\sum_{t=1}^n \ln \big( (2\pi e)^{n_y} \det\big(C_t \Sigma_{t} C_t^T + N_t K_{W_t} N_t^T\big) \big) \label{entr_noise}\\
& \hat{I}_t \in G(0, K_{\hat{I}_t}) \hso  \mbox{an orth. innov. proc.  indep. of  $V^{t-1}$},  \label{inn_po_2}\\
&K_{\hat{I}_t} \tri  cov(\hat{I}_t, \hat{I}_t)= C_t \Sigma_t C_t^T +N_t K_{W_t} N_t^T=K_{V_t|V^{t-1}}. \label{cov_in_noise}
\end{align}
 where  $\Sigma_t $ satisfies the generalized matrix DRE 
\begin{align}
\Sigma_{t+1}= &A_t \Sigma_{t}A_t^T  + B_tK_{W_t}B_t^T -\Big(A_t \Sigma_{t}C_t^T+B_tK_{W_{t}}N^T  \Big) \nonumber \\
& \hspace{0.1cm} . \Big(N_t K_{W_t} N_t^T+C_t  \Sigma_{t} C_t^T\Big)^{-1}\Big( A_{t}  \Sigma_{t}C_t^T+ B_t K_{W_t}N_t^T  \Big)^T,\nonumber\\   & \hspace{0.1cm}\hso \Sigma_t \succeq 0, \hso  t=1, \ldots, n, \hso \Sigma_{1}=K_{S_1}\succeq 0. \label{dre_1}
\end{align}
(v) An equivalent characterization of ${C}_n(\kappa, {\bf P}_{Y_1})$ is 
\begin{align}
 &{C}_n(\kappa, {\bf P}_{Y_1})
=  \sup_{ \frac{1}{n}  {\bf E}\big\{\sum_{t=1}^n  ||X_t||_{{\mathbb R}^{n_x}}^2\big\}\leq \kappa } \frac{1}{2} 
\sum_{t=1}^n \ln \Big\{ \frac{ \det( K_{I_t})   }{\det(K_{\hat{I}_t})}\Big\}^+  
\end{align}
  where the supremum is over $(F_t, G_t, \Gamma_t, D_t, K_{Z_t}), t=1, \ldots, n$. 
 \end{theorem}

\label{sect_POSS}

\subsection{Asymptotic Analysis}
\label{sect:q1}
To address the limit $\lim_{n \longrightarrow \infty}\frac{1}{n} C_{n}(\kappa, {\bf P}_{V_1})$,  under Case 1 and Case 2,  we investigate the convergence properties of generalized matrix DREs (\ref{DREy}), (\ref{dre_1}) and Lyapunov matrix difference equation (\ref{lyapunov}) to their limits.  Such properties are  summarized in \cite[Theorem A.1]{charalambous-kourtellaris-loykaIT2015_Part_2} and \cite[Theorem III.2]{charalambous2020new}.

We present sufficient conditions for convergence in Corollary~\ref{cor_POSS_IH},  Theorem~\ref{thm_fc_IH}, Theorem~\ref{lemma_vanschuppen}, irrespectively of whether the noise is stable, or unstable.  

\begin{corollary}
\label{cor_POSS_IH}
Consider  Case 1 or Case 2 \\
Let  $\Sigma_t, t=1, 2, \ldots $ denote the solution of the matrix DRE (\ref{dre_1}). \\
Let $\Sigma=\Sigma^{T} \succeq 0
$ be a solution of the corresponding ARE  
\begin{align}
\Sigma= &A \Sigma A^T  + B K_{W}B^T -\Big(A \Sigma C^T+B K_{W}N^T  \Big) \nonumber \\&\hspace{0.1cm} . \Big(N K_{W} N^T+C  \Sigma C^T\Big)^{-1} \Big( A  \Sigma C^T+ B K_{W}N^T  \Big)^T.  \label{dre_1_SS}
\end{align}
Define the matrices 
\begin{align}
&A^*\tri A- B K_W N^T \big(N K_W N^T\big)^{-1} C, \hso G\tri B, \nonumber \\& B^* \tri K_W- K_W N^T \Big(N K_W N^T\Big)^{-1} \Big(K_W N^T\Big)^T. \label{noise_st}
\end{align}
Suppose (see \cite{kailath-sayed-hassibi,caines1988,vanschuppen2021} for definitions)
\begin{align}
&\mbox{$ \{A, C\}$ is detectable,  and $\{A^*, G B^{*,\frac{1}{2}}\}$ is stabilizable.} \label{st}
\end{align}
Any solution $\Sigma_{t}, t=1, 2, \ldots,n$ to the generalized matrix DRE (\ref{dre_1})  with arbitrary  initial condition $\Sigma_{1}\succeq 0$, is such that $\lim_{n \longrightarrow \infty} \Sigma_{n} =\Sigma$, where $\Sigma\succeq  0$ is the unique solution of the generalized matrix ARE (\ref{dre_1_SS}) with $spec\big(M^{CL}(\Sigma)\big) \in {\mathbb D}_o$.   
\end{corollary}
\vspace*{-.3cm}
\begin{proof} For Case 1, the convergence of $\Sigma_n, n=1,2, \ldots$, follows from the detectability and stabilizability conditions.  For Case 2, the statements of convergence of $\Sigma_n, n=1,2, \ldots$  hold, due to continuity property of solutions of generalized difference Riccati equations, with respect to its coefficients. 
\end{proof}
\vspace{-0.5cm}
\begin{theorem} 
\label{thm_fc_IH}
Consider  Case 1 or Case 2.  \\
Let  $\Pi_t, t=1, \ldots, $ denote the solution of the DRE (\ref{DREy}). \\
Let $\Pi = \Pi^{T} \succeq 0$ be a solution of the corresponding ARE 
\begin{align}
&\Pi= {\bf A} \Pi{\bf A}^T  + {\bf B}K_{\overline{W}}{\bf B}^T -\Big({\bf A}  \Pi{\bf C}^T+{\bf B}K_{\overline{W}}{\bf D}^T  \Big) \nonumber \\
& \hspace{0.1cm} . \Big({\bf D} K_{\overline{W}} {\bf D}^T+{\bf C}  \Pi {\bf C}^T\Big)^{-1}\Big( {\bf A}  \Pi {\bf C}^T+ {\bf B} K_{\overline{W}}{\bf D}^T  \Big)^T.\label{kf_m_4_a_TI_ARE}
\end{align}
Define the matrices \cite{kailath-sayed-hassibi,caines1988,vanschuppen2021}
\begin{align}
&{\bf A}^*\tri {\bf A} - {\bf B} K_{\overline W} {\bf D}^T \big({\bf D} K_{\overline W} {\bf D}^T\big)^{-1} {\bf C}, \hso {\bf G}\tri  {\bf B}, \nonumber \\& {\bf B}^* \tri K_{\overline W}- K_{\overline W} {\bf D}^T \Big({\bf D} K_{\overline W} {\bf D}^T\Big)^{-1} \Big(K_{\overline W} {\bf D}^T\Big)^T. \label{input_st}
\end{align}
Suppose \cite{kailath-sayed-hassibi,caines1988,vanschuppen2021}
\begin{align}
&\mbox{$\{{\bf A}, {\bf C}\}$ is detectable and $\{{\bf A}^*, {\bf G} {\bf B}^{*,\frac{1}{2}}\}$ is stabilizable.} \label{kt}
\end{align}
Any solution $\Pi_{t}, t=1, 2, \ldots,n$ to the generalized matrix DRE (\ref{DREy}) with arbitrary initial condition $\Pi_{1}\succeq 0$, is such that $\lim_{n \longrightarrow \infty} \Pi_{n}=\Pi$, where $\Pi\succeq  0$ is the unique solution of the generalized matrix ARE (\ref{kf_m_4_a_TI_ARE}), with $spec\big({\bf F}^{CL}(\Pi)\big) \in {\mathbb D}_o$. 
\end{theorem}
\begin{proof} Similar to Corollary~\ref{cor_POSS_IH}.
\end{proof}
Theorem~\ref{lemma_vanschuppen} identifies conditions for the average power (\ref{cp_9_al_n_nfb_a}) to converge,  using  $P_t =  cov\big(\Xi_t,\Xi_t)$, which satisfies (\ref{lyapunov}). 
\begin{theorem} Convergence of average power\\
\label{lemma_vanschuppen}
Consider the average power of Thm~\ref{thm_SS_nfb}, for Cases 1 or 2. \\
Let  $P_t,  t=1 \ldots, n$ be a solution of Lyapunov recursion (\ref{lyapunov}). \\
Let  $P\succeq 0$ be a solution of 
\begin{align}
P=F  P F^T +G K_{Z}  G^T. \label{MP_1_new}
\end{align}
Suppose $F$ is an exponentially stable matrix.  Any solution $P_{t}, t=1, 2, \ldots,n$ to the Lyapunov recursion  DRE (\ref{lyapunov}), with arbitrary initial condition $P_{1}\succeq 0$, is such that $\lim_{n \longrightarrow \infty} P_{n}=P$, where $P\succeq  0$ is the unique solution of (\ref{MP_1_new}). 
Moreover, 
 \begin{align}
\lim_{n \longrightarrow \infty}& \frac{1}{n}  {\bf E} \Big\{\sum_{t=1}^{n} ||X_t||_{{\mathbb R}^{n_x}}^2\Big\}= \lim_{n \longrightarrow \infty} \frac{1}{n}  \sum_{t=1}^{n}  tr \Big(\Gamma P_t \Gamma^T +D K_{Z}D^T\Big) \nonumber  \\
=& tr \Big(\Gamma P \Gamma^T +D K_{Z}D^T\Big), \; \forall P_1\succeq 0.   \label{cp_9_al_n_nfb_a2}
\end{align}
\end{theorem}
\begin{proof} For Case 1, these are known\cite{vanschuppen2021}.  For Case 2,  the statements are due to continuity property of solutions of Lyapunov equations, with respect to their coefficients.
\end{proof}
\subsection{Asymptotic Characterizations of Nonfeedback Capacity }
\label{sect:cor-solu}
 \begin{theorem} Characterization of $C^\infty(\kappa,{\bf P}_{Y_1})$ for Case 1 \\ 
\label{prob_1_in}
Consider the time-invariant noise and channel input strategies of Case 1, i.e., (\ref{tic_1}) and (\ref{tic_2})  hold. \\
Define the per unit time limit and supremum by\footnote{If at any time $t$, the information $H(Y_t|Y^{t-1})-H(V_t|V^{t-1})=+\infty$ then it is removed,  as it is usually the case \cite{pinsker1964}.}
\begin{align}
C^\infty(\kappa,{\bf P}_{Y_1}) \tri & \sup_{{\cal P}_{\infty}(\kappa)}\lim_{n \longrightarrow\infty}\frac{1}{2n}          \sum_{t=1}^n     \ln   \Big\{     \frac{ \det \big( {\bf C} \Pi_t {\bf C}^T + {\bf D} K_{\overline{W}}{\bf D}^T \big)    }{\det\big(C \Sigma_t C^T +N K_{W} N^T\big)}\Big\}^+ \label{i_ll_2}
\end{align}
where the average power constraint is defined by 
\begin{align}
&{\cal P}_{\infty}(\kappa)\tri   \Big\{(F, G, \Gamma, D, K_{Z}) \in {\cal P}^\infty \Big| \\
&\hspace*{1.cm} \lim_{n \longrightarrow \infty}\frac{1}{n}  \sum_{t=1}^n tr (\Gamma P_t {\Gamma}^T +D K_Z {D}^T)  \leq \kappa\Big\},
 \label{Q_1_10_s1_new}\\
& {\cal P}^\infty \tri  \Big\{(F, G, \Gamma, D, K_{Z}), \mbox{ such that the following hold}\nonumber \\&\mbox{(i) the detectability and stabilizability of (\ref{st})},\\
& \mbox{(ii) the detectability and stabilizability of (\ref{kt})}\\
& \mbox{(iii) $F$ is exponentially stable}\Big\}.
 \label{adm_set}
 \end{align}
%
Then, $C^\infty(\kappa,{\bf P}_{Y_1})$  is given by 
\begin{align}
 C^\infty(\kappa,{\bf P}_{Y_1}) & =  \sup_{ {\cal P}^\infty(\kappa)} \frac{1}{2} \ln \Big\{ \frac{ \det \big( {\bf C} \Pi {\bf C}^T + {\bf D} K_{\overline{W}}{\bf D}^T \big)    }{\det\big(C \Sigma C^T +N K_{W} N^T\big)}\Big\}^+ \nonumber \\& \tri C^{\infty}(\kappa),\hso  \forall {\bf P}_{Y_1} \label{ll_3}
\end{align}
where ${\cal P}^\infty(\kappa)$ is defined by 
\begin{align}
{\cal P}^\infty(\kappa)&\tri  \Big\{(F, G, \Gamma, D, K_{Z})\in {\cal P}^\infty\Big| \nonumber \\ &K_Z\succeq 0, \hso tr( \Gamma P{\Gamma}^T +D K_ZD^T) \leq \kappa 
 \Big\} \label{ll_4}
 \end{align}
and $\Sigma \succeq 0$ and $\Pi \succeq 0$ are the unique and stabilizable solutions,  i.e.,  $spec\big(M^{CL}(\Sigma)\big) \in {\mathbb D}_o$ and $spec\big({\bf F}^{CL}(\Pi)\big) \in {\mathbb D}_o$ of the generalized matrix AREs (\ref{dre_1_SS}) and (\ref{kf_m_4_a_TI_ARE}) respectively,   $P\succeq  0$ is the  unique solution of the matrix Lyapunov equation (\ref{MP_1_new}),
provided there exists $\kappa\in [0,\infty)$, such that ${\cal P}^\infty(\kappa)$ is non-empty.\\
Moreover, the optimal $(F, G, \Gamma, D, K_{Z}) \in {\cal P}^\infty(\kappa)$, is such that, \\
 (i) if the noise is stable,  then the input and the output processes $(X_t, Y_t), t=1, \ldots$ are asymptotic stationary and \\
 (ii) if the noise is unstable,  then the input and the innovations processes $(X_t, I_t), t=1, \ldots$ are asymptotic stationary. 
\end{theorem} 
\vspace{-0.5cm}
 \begin{proof} By the definition of the set ${\cal P}^\infty$, then 
 Corollary~\ref{cor_POSS_IH}, Theorem~\ref{thm_fc_IH} and Theorem~\ref{lemma_vanschuppen} hold. 
Hence, the following summands converge in $[0,\infty)$ uniformly, $\forall {\bf P}_{Y_1}$:
 \begin{align}
&\lim_{n \longrightarrow\infty}\frac{1}{n}
 \sum_{t=1}^n tr(\Gamma P_t \Gamma^T +D K_Z D^T) =  tr(\Gamma P \Gamma^T +D K_Z D^T), \label{conve_1} \\
 & \lim_{n \longrightarrow\infty}\frac{1}{2n}    \Big\{  \sum_{t=1}^n  \ln \big( \frac{ \det \big( {\bf C} \Pi_t {\bf C}^T + {\bf D} K_{\overline{W}}{\bf D}^T \big)    }{\det\big(C \Sigma_t C^T +N K_{W} N^T\big)}\big)\Big\} \nonumber \\ 
 &=\frac{1}{2} \ln \big( \frac{ \det \big( {\bf C} \Pi {\bf C}^T + {\bf D} K_{\overline{W}}{\bf D}^T \big)    }{\det\big(C \Sigma C^T +N K_{W} N^T\big)}\big), \hso   \forall {\bf P}_{Y_1} .  \label{conve_2}
 \end{align} 
 The last part of the theorem follows from the asymptotic properties of the Kalman-filter, as follows.  For (i).  $\Xi_{t}, t=1, \ldots$ is asymptotically stationary, which implies $X_t= \Gamma \Xi_{t} +  D Z_t,  X_1 = \Gamma {\Xi}_1 + D Z_1,  t=2, \ldots, n,
 Z_t \in G(0, K_{Z}), \hso K_Z\succeq 0$,  $I_t, t=1, \ldots$ and $Y_t=HX_t+ V_t, t=1, \ldots$ are   asymptotically stationary.  Similarly for (ii), with the exception that $Y_t=H X_t+ V_t, t=1, \ldots$ is not asymptotically stationary, because $V_t, t=1, \ldots $ is unstable.
 \end{proof}

Next, we show that Theorem~\ref{prob_1_in} remains valid for Case 2.

\begin{corollary} Characterization of $C^\infty(\kappa,{\bf P}_{Y_1})$ for Case 2 \\ 
 \label{the_limsup}
Consider the asymptotically time-invariant noise and channel input strategies of Case 2, i.e., (\ref{atic_1}) and  (\ref{atic_2})  hold.\\
Define the per unit time limit and supremum by 
\begin{align}
&C^{\infty,+}(\kappa,{\bf P}_{Y_1}) \tri  \sup_{{\cal P}_{\infty}^{+}(\kappa) }\lim_{n \longrightarrow\infty}\frac{1}{2n}  \Big\{       \nonumber \\
& \hspace*{1.0cm} \sum_{t=1}^n \ln \Big\{ \frac{ \det \big( {\bf C}_t \Pi_t {\bf C}_t^T + {\bf D}_t K_{\overline{W}_t}{\bf D}_t^T \big)    }{\det\big(C_t \Sigma_t C_t^T +N_t K_{W_t} N_t^T\big)}\Big\}^+\Big\} \label{i_ll_2_ati}\\
&{\cal P}_{\infty}^+(\kappa)\tri   \Big\{\{(F_n, G_n, \Gamma_n, D_n K_{Z_n})| n=1,2, \ldots \} \in  {\cal P}_{\infty}^+\Big|   \nonumber \\
&\hspace*{1.0cm}  \lim_{n \longrightarrow\infty}\frac{1}{n}  \sum_{t=1}^n tr\big(\Gamma_t P_t \Gamma_t^T +D_t K_{Z_t} D_t^T\big)\big\} \leq \kappa\Big\},   \label{Q_1_10_s1_new_ati} \\
  & {\cal P}_\infty^+ \tri  \Big\{\{ (F_n, G_n, \Gamma_n, D_n, K_{Z_n}) |n = 1,2,\dots\}\Big| \nonumber \\ 
  &  \lim_{n \longrightarrow \infty} (F_n, G_n, \Gamma_n, D_n K_{Z_n})=(F, G, \Gamma, D, K_{Z})\in    {\cal P}^\infty    \Big\}.  \label{adm_set_ati}
 \end{align}
Then, 
\begin{align}
 C^{\infty,+}(\kappa,{\bf P}_{Y_1})=C^{\infty}(\kappa,{\bf P}_{Y_1})=C^{\infty}(\kappa)=\mbox{(\ref{ll_3})}, \hso \forall {\bf P}_{Y_1}.
 \end{align}
 and the statements of Theorem~\ref{prob_1_in}.(i), (ii), remain valid.
\end{corollary} 
\begin{proof} 
The solutions of the DREs and the Lyapunov equation are,  $\Sigma_{n+1}= \Sigma_{n+1}(\Sigma_n,A_n,B_n,C_n,N_n,K_{W_n})$, $\Pi_{n+1}=\Pi_{n+1}(\Pi_{n},\Sigma_{n},{\bf A}_n, {\bf B}_n,{\bf C}_n,{\bf D}_n, K_{\overline{ W}_n}), $\\$ P_{n+1} = P_{n+1}(P_n,F_n,G_n,\Gamma_n,D_n, K_{Z_n}),\;n=1,2, \ldots$ and these are continuous with respect to their coefficients. Moreover, for all elements of the set ${\cal P}^\infty$,  by (\ref{atic_1}), then 
 \begin{align}
&\lim_{n \longrightarrow\infty}\frac{1}{n}\sum_{t=1}^n tr (\Gamma_t P_t {\Gamma_t}^T +D_t K_{Z_t} {D_t}^T) =tr( \Gamma P \Gamma^T +D K_Z D^T), \; \forall {\bf P}_{Y_1},     \label{conve_1_tv} \\
 & \lim_{n \longrightarrow\infty}\frac{1}{2n}  \sum_{t=1}^n \ln \Big\{ \frac{ \det \big( {\bf C}_t \Pi_t {\bf C}_t^T + {\bf D}_t K_{\overline{W}_t}{\bf D}_t^T \big)    }{\det\big(C_t \Sigma_t C_t^T +N_t K_{W_t} N_t^T\big)}\Big\}^+ \nonumber \\
 &=\frac{1}{2} \ln \big( \frac{ \det \big( {\bf C} \Pi {\bf C}^T + {\bf D} K_{\overline{W}}{\bf D}^T \big)    }{\det\big(C \Sigma C^T +N K_{W} N^T\big)}\big), \hso \forall {\bf P}_{Y_1}.   \label{conve_2_tv}
 \end{align} 
The rest follows by repeating the proof of Theorem~\ref{prob_1_in}.
\end{proof} 

Identity $C^{o}(\kappa,{\bf P}_{Y_1}) =C^{\infty}(\kappa,{\bf P}_{Y_1}) = C^{\infty}(\kappa), \forall {\bf P}_{Y_1}$ for Case 2, follows from the uniform convergence of Theorem~\ref{prob_1_in} and Corollary~\ref{the_limsup}; the derivation is omitted due to space limitation.

\begin{theorem} Characterization of $C^o(\kappa,{\bf P}_{Y_1})$ for Case 2 \\ 
 \label{thm:limsup_ati}
Consider the asymptotically time-invariant noise and channel input strategies of Case 2,  i.e., (\ref{atic_1}) and (\ref{atic_2}) hold. \\
Define the per unit time limit and supremum by 
\begin{align}
&C^{o}(\kappa,{\bf P}_{Y_1}) \tri \lim_{n \longrightarrow\infty} \sup_{{\cal P}_{n}^{o,+}(\kappa)}\frac{1}{2n} \Big\{ \nonumber \\
&\hspace*{1.0cm} \sum_{t=1}^n \ln \Big\{ \frac{ \det \big( {\bf C}_t \Pi_t {\bf C}_t^T + {\bf D}_t K_{\overline{W}_t}{\bf D}_t^T \big)    }{\det\big(C_t \Sigma_t C_t^T +N_t K_{W_t} N_t^T\big)}\Big\}^+\Big\} \label{limsup_ati_1}\\
&{\cal P}_{\infty}^+(\kappa)\tri   \Big\{\{(F_n, G_n, \Gamma_n, D_n K_{Z_n})| n=1,2, \ldots \} \in  {\cal P}_{\infty}^+\Big|   \nonumber \\
&\hspace*{1.0cm} \frac{1}{n} \sum_{t=1}^n tr(\Gamma_t P_t \Gamma_t^T +D_t K_{Z_t} D_t^T)\big\} \leq \kappa\Big\}. \label{limsup_ati_2} 
\end{align}
Then, 
\begin{align}
 C^{o}(\kappa,{\bf P}_{Y_1})=C^{\infty}(\kappa,{\bf P}_{Y_1})=C^{\infty}(\kappa)=\mbox{(\ref{ll_3})}, \: \forall {\bf P}_{Y_1}
 \end{align}
 and the statements of Theorem~\ref{prob_1_in}.(1), (ii),  remain valid.
\end{theorem} 
\begin{proof} The derivation uses the uniform limits (\ref{conve_1_tv}) and (\ref{conve_2_tv}), Theorem~\ref{prob_1_in} and Corollary~\ref{the_limsup}.
\end{proof}

\begin{remark} 
If the {\it stabilizability condition} is replaced by the {\it unit circle controllability} (see \cite{kailath-sayed-hassibi} for definition and \cite{charalambous2020new,louka-kourtellaris-charalambous:2020b,kourtellaris-charalambous-loyka:2020b,charalambous-kourtellaris-loykaIEEEITC2019} for specific examples), then Theorem~\ref{prob_1_in} and Theorem~\ref{thm:limsup_ati} remain valid with the fundamental difference that all limits are not uniform  for all ${\bf P}_{V_1}$. For such a relaxation, the limits depend on ${\bf P}_{V_1}, \Sigma_1, P_1,  \Pi_1$ and the asymptotic optimization problem $C^\infty(\kappa)$ may not be convex. Specific examples are found in \cite{louka-kourtellaris-charalambous:2020b}. 
\end{remark}
 
\section{Conclusion}
This paper presents new asymptotic characterizations of nonfeedback capacity of MIMO additive Gaussian noise (AGN) channels,  when the noise is nonstationary and unstable. The asymptotic characterizations of nonfeedback capacity,  involve two generalized matrix algebraic Riccati equations (AREs) of filtering theory and a Lyapunov matrix equation of stability theory of Gaussian systems. Identified, are conditions for uniform convergence of the asymptotic limits, which imply that the nonfeedback capacity is independent of the initial states. 



\newpage 

\bibliographystyle{IEEEtran}

\bibliography{Bibliography_capacity}

\end{document}